\documentclass[journal]{IEEEtran}
\usepackage{amssymb}
\usepackage{amsmath}
\usepackage{cite}
\usepackage{algorithm}
\usepackage{algorithmic}
\usepackage[dvips]{graphicx}
\usepackage{stfloats}
\begin{document}
%
\title{Efficient Beamforming for MIMO Relaying Broadcast Channel with Imperfect Channel Estimation}

\author{Zijian Wang, Wen Chen,~\IEEEmembership{Senior Member,~IEEE} and Jun Li,~\IEEEmembership{Member,~IEEE}
\thanks{Copyright (c) 2010 IEEE. Personal use of this material is permitted. However,
permission to use this material for any other purposes must be
obtained from the IEEE  by sending a request to
pubs-permissions@ieee.org.}

\thanks{Z. Wang is with Department of Electronic Engineering, Shanghai
Jiao Tong University, China, and the State Key Laboratory of
Integrated Services Networks, Xidian University, Xi¡¯an 710071
(e-mail: wangzijian1786@sjtu.edu.cn).}

\thanks{W. Chen is with Department of Electronic Engineering, Shanghai Jiao Tong University, China, and
the State Key Laboratory for Mobile Communications, Southeast
University, Nanking, China (e-mail: wenchen@sjtu.edu.cn).}

\thanks{ J. Li is with the
School of Electrical Engineering and Telecommunications, University
of New South Wales, Australian (email: jun.li@unsw.edu.au).}

\thanks{This work is supported by NSF China~\#60972031,
by national 973 project~\#2012CB316106 and~\#2009CB824900, by
national huge special project~\#2012ZX03004004, by national key
laboratory project~\#W200907 and~\#ISN11-01, by Huawei Funding
\#YBWL2010KJ013.}
}

\markboth{IEEE Transactions on Vehicular Technology.}{}
\maketitle

\begin{abstract}
We consider a multiple-input multiple-output (MIMO) relaying
boardcast channel in downlink cellular networks,  where the base
station and the relay stations are both equipped with multiple
antennas, and each user terminal has only a single antenna. In
practical scenarios, channel estimation is imperfect at the
receivers. Aiming at maximizing the SINR at each user, we develop
two robust linear beamforming schemes respectively for the single
relay case and  the multi-relay case. The two proposed schemes are
based on sigular value decomposition (SVD), minimum mean square
error (MMSE) and regularized zero-forcing (RZF). Simulation results
show that the proposed scheme outperforms the conventional schemes
with imperfect channel estimation.

\end{abstract}

\begin{keywords}
MIMO relaying broadcast, MMSE receiver, RZF precoding, SINR, Singular value
decomposition.
\end{keywords}

\section{Introduction}

In recent years, MIMO relay networks have drawn considerable
interest  due to the advantages to increase the data rate and extend
coverage in the cellular edge. The MIMO relay network with perfect
channel state information (CSI) have been studied in~\cite{9,19}. In
\cite{9}, the authors investigate the linear processing at relay for
MIMO relay networks with fairness requirement. In~\cite{19}, the
authors investigate the regularized zero-forcing (RZF) precoder at
relays, which is observed to have an advantage to the zero-forcing
(ZF) and the matched filter (MF) precoders. But the RZF precoder is
not optimized and constantly chooses one as the regularing factor.
The MIMO relaying broadcast network has been considered
in~\cite{21}, where the singular value decomposition (SVD) and ZF
precoder are respectively used to the backward channels (BC) and the
forward channels (FC) to optimize the joint precoding. The authors
use an iterative method to show that the optimal precoding matrices
always diagonalize the compound channel of the system.

All the above works consider perfect CSIs. However, perfect CSI is
usually difficult to be obtained for a practical system.
In~\cite{25}, MMSE based precoding has been considered in multiple
antenna broadcast channel with  imperfect CSI at the source.
In~\cite{wang}, the authors optimized a QR based beamformings with imperfect $\mathcal{R}$-$\mathcal{D}$ CSI due to large delay.
 Works for limited feedback in MIMO relay networks are studied in~\cite{23,30}, and in MIMO relaying broadcast
channel are studied in~\cite{22,26,A}.
 In~\cite{22}, the authors further
study the impact of feedback bits of BC and FC on the
achievable rates for the linear processing scheme in~\cite{21}.
 In
~\cite{26}, based on MMSE criteria, robust ZF precoding are
considered at the relay using the limited feedback of CSI to the
relay. But only imperfect forward channel (FC) is considered.
In~\cite{A}, the authors propose an MMSE based beamforming design in
a MIMO relay broadcast channel with finite rate feedback.

In this paper, we study MIMO relaying downlink broadcast channel in
a  wireless cellular network. Focusing on linear beamformings, we propose a robust beamforming scheme considering both imperfect channel estimation at relay and user terminals.
The proposed scheme is based on SVD-RZF for the single relay case and  MMSE-RZF for the multi-relay case. By maximizing
the derived signal-to-interference noise ratio (SINR), we optimize
the MMSE receiver and RZF precoder.  Simulation results show
that the proposed robust SVD-RZF and MMSE-RZF outperform
other conventional beamformers.

In this paper, boldface lowercase letter and boldface uppercase
letter represent vectors and matrices, respectively. Notation
$\mathbb{C}^N$ denotes an $N\times 1$ complex vector. The
$\mathrm{{tr}}(\mathbf{A})$ and $\mathbf{A}^H$ denote the trace and
the conjugate transpose of a matrix $\mathbf{A}$, respectively.
$(\mathbf{a})_k$ and $(\mathbf{A})_{j,k}$ represent the $k$-th entry
of vector $\mathbf{a}$ and the $(j,k)$-th entry of matrix
$\mathbf{A}$ respectively. $\mbox{\boldmath $\mathbf{I}$}_N$ denotes
the $N{\times}N$ identity matrix.
Finally, we denote the expectation operation by $\rm
E\{\cdot\}$.


\section{System Model}

We consider a MIMO relaying broadcast network which consists a base
station, $R$ fixed relays, and $K$ user terminals as depicted in
Fig.~1. The base station is equipped with $M$ antennas, each relay
is equipped with $N$ antennas and each user terminal only has a
single antenna. It is supposed that $M,N\geq K$ so that the network
can support $K$ independent data streams.  A broadcast transmission
is composed of two phases. During the first phase, the base station
broadcasts $M$ precoded data streams to the relays after applying a
linear precoder to the original data vector ${\bf{s}}\in
\mathbb{C}^K$, where $\mathrm{E} \{{\bf{s}}{\bf{s}}^H
\}={\bf{I}}_K$. We denote the precoding matrix at the base station
as ${\bf{F}}$ and suppose that the base station transmit power is
$P_s$.  Because we have $\mathrm{E} \{ {\bf{s}}^H {\bf{F}}^H
{\bf{F}} {\bf{s}}\}=\mathrm{tr}({\bf{F}}^H {\bf{F}})$, the power
control factor at the base station is
$\rho_s=\sqrt{\frac{P_s}{\mathrm{tr}({\bf{F}}^H {\bf{F}})}}$. The
received signal vector at the $r$-th relay is
\begin{equation}
{\bf{y}}_r=\rho_s {\bf{H}}_r {\bf{F}} {\bf{s}}+{\bf{n}}_r,
\end{equation}
where ${\bf{H}}_r\in  \mathbb{C}^{N\times M}$ is the Rayleigh BC
matrix of the $r$-th relay, in which, all entries are  $i.i.d$ complex Gaussian
distributed with zero mean and unit variance, and ${\bf{n}}_r\in
\mathbb{C}^N$  is the noise vector at the relay, in which, all the
entries are $i.i.d$ complex Gaussian distributed with zero mean and
variance $\sigma_1^2 {\bf{I}}_N$.
During the second phase, the relays all broadcast the signal vector to
the user terminals after a precoding matrix ${\bf{W}}_r$. The
transmit power at  the relay is $P_r$, and the power control factor
is $\rho_r$, where
\begin{equation}
\rho_r=\left(\frac{P_r}{\mathrm{tr}(\rho_s^2
\mathbf{W}_r{\mathbf{H}}_r
\mathbf{F}\mathbf{F}^H{{\mathbf{H}}}_r^H\mathbf{W}_r^H+\sigma_1^2\mathbf{W}_r\mathbf{W}_r^H)}
\right)^{\frac{1}{2}}. \label{8}
\end{equation}
\begin{figure}[!t]
\centering
\includegraphics[width=3in]{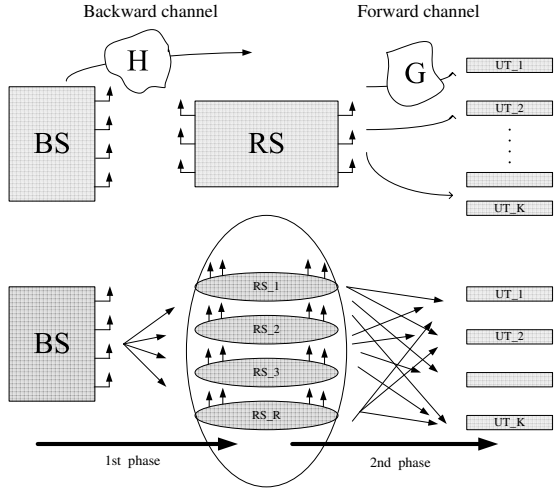}
\caption{MIMO relay broadcast channel with imperfect channel estimation.}
\end{figure}
Denoting the received signal at the $k$th user terminal as
${y}_k$, the received vector at user terminals can thus be written
as
\begin{equation}
\begin{split}
{\bf{y}} =& \left [ y_1, y_2,\ldots,y_K \right ]\\=&\sum_{r=1}^R\rho_r{\bf{G}}_r{\bf{W}}_r\left(\rho_s{\bf{H}}_r{\bf{F}} {\bf{s}}+{\bf{n}}_r \right)+{\bf{n}}_D
,
\label{4}
\end{split}
\end{equation}
where ${\bf{n}}_D\in \mathbb{C}^K$ denotes the noise vector at the
user terminals,  in which, all entries are $i.i.d$ Gaussian
distributed with zero mean and $\sigma_2^2$ variance, ${\bf{G}}_r$ is
the is the Rayleigh FC matrix of the $r$-th relay.

Considering imperfect channel estimation at both the relay and user
terminals,  we model the channel state information (CSI) as
\begin{equation}
{\bf{H}}_r=\widehat{\bf{H}}_r+e_1{\bf{\Omega}}_{1,r},
\end{equation}
and
\begin{equation}
\left [ {\mathbf{g}}_{1,r}^H,
{\mathbf{g}}_{2,r}^H,\ldots,  {\mathbf{g}}_{K,r}^H \right]^H={\bf{G}}_r=\widehat{\bf{G}}_r+e_2{\bf{\Omega}}_{2,r},
\end{equation}
where $\mathbf{g}_{k,r}^H\in \mathbb{C}^N$ is the CSI of the $r$-th relay to the $k$-th
user channel. The entries of ${\bf{\Omega}}_{1,r}$ and ${\bf{\Omega}}_{2,r}$
are $i.i.d$ complex Gaussian distributed with zero mean and unit
variance. $\widehat{\bf{H}}_r$ and $\widehat{\bf{G}}_r$ are the
estimated CSIs and they are respectively independent of
${\bf{\Omega}}_{1,r}$ and ${\bf{\Omega}}_{2,r}$. $e_1^2$ and $e_2^2$ denotes
the channel estimation error powers. We suppose that each user
has the same channel estimation error power for simplicity.
\section{SINR at User Terminals}
Considering channel estimation errors, (\ref{4}) becomes
\begin{multline}
\mathbf{y}=\sum_{r=1}^R\rho_s\rho_r\widehat{\mathbf{G}}_r\mathbf{W}_r\widehat{\mathbf{H}}_r\mathbf{F}\mathbf{s}\\+\sum_{r=1}^R
\rho_s\rho_r\left(e_1 \widehat{\mathbf{G}}_r\mathbf{W}_r\mathbf{\Omega}_{1,r}\mathbf{F}
+e_2\mathbf{\Omega}_{2,r}\mathbf{W}_r\widehat{\mathbf{H}}_r\mathbf{F}\right)\mathbf{s}\\+\sum_{r=1}^R
\rho_r\left(\widehat{\mathbf{G}}_r+e_2\mathbf{\Omega}_{2,r}\right)\mathbf{W}_r\mathbf{n}_r+\mathbf{n}_D,
\label{3}
\end{multline}
where we omitted the term involving $e_1e_2$ because we assume $e_1,e_2\ll1$. We can write (\ref{3}) as
\begin{equation}
\mathbf{y}=\mathbf{H}_{\mathrm{eff}}\mathbf{s}+\mathbf{n},
\end{equation}
where $\mathbf{H}_{\mathrm{eff}}\mathbf{s}$ is the first term and $\mathbf{n}$ is the rest terms in the right-hand-side of (\ref{3}). Then the SINR at the $k$-th user terminal can be calculated by
\begin{equation}
\mathrm{SINR}_k=\frac{{|(\mathbf{H}_{\mathrm{eff}})_{k,k}|}^2}{\sum\limits_{j=1,j\neq k}^K{|(\mathbf{H}_{\mathrm{eff}})_{k,j}|}^2+\mathrm{E}\{\mathbf{n}_k\mathbf{n}_k^{*}\}},
\end{equation}
where
\begin{multline}
\mathrm{E}\{\mathbf{n}_k\mathbf{n}_k^{*}\}
\\=\sum_{r=1}^R\left(\frac{e_1^2\rho_s^2\rho_r^2}{K}\mathrm{tr}(\mathbf{F}\mathbf{F}^H)
\mathrm{tr}(\widehat{\mathbf{G}}_r\mathbf{W}_r\mathbf{W}_r^H{\widehat{\mathbf{G}}_r}^H)\right.
\\+e_2^2\rho_s^2\rho_r^2\mathrm{tr}(\mathbf{W}_r\widehat{\mathbf{H}}_r\mathbf{F}\mathbf{F}^H{\widehat{\mathbf{H}}_r}^H\mathbf{W}_r^H)
\\+\frac{\rho_r^2\sigma_1^2}{K}\mathrm{tr}(\widehat{\mathbf{G}}_r\mathbf{W}_r\mathbf{W}_r^H{\widehat{\mathbf{G}}_r}^H)
\\\left.+\rho_r^2 e_2^2\sigma_1^2\mathrm{tr}(\mathbf{W}_r\mathbf{W}_r^H)\right)+\sigma_2^2.
\label{10}
\end{multline}
In the derivation, we used the fact  $\mathrm{E}\{\mathbf{\Omega} \mathbf{A} \mathbf{\Omega}^H\}=\mathrm{tr}(\mathbf{A})\mathbf{I}_N$ for any $N\times N$ matrix $\mathbf{A}$.
The expectation is taken over all distributions of $\mathbf{s},\mathbf{n}_r,\mathbf{n}_D,\mathbf{\Omega}_{1,r},\mathbf{\Omega}_{2,r}$. Our aim is to find the precoding matrix at the base station and the beamforming matrix at the relay to maximize the SINR at each user terminal.

It is difficult to directly obtain the optimum closed-form solution
because the  optimization problem is not convex. In fact, there is no
optimal beamforming even for perfect CSI in the MIMO relaying
broadcast channels~\cite{22}. In the following section, we propose a
robust beamforming scheme for two different cases which considers the imperfect channel
estimations.
\section{Robust Beamforming Design}
\subsection{SVD-RZF based design for the single relay case}
If there is only one relay, for the first phase, the transmission is similar to a point-to-point
MIMO system. Therefore, we propose an SVD-based beamforming for the
backward channel~\cite{1}. Using singular value decomposition (SVD),
the imperfect BC matrix can be decomposed as
\begin{equation}
\widehat{\bf{H}}=\widehat{\bf{H}}_1={\bf{U}}{\bf{\Sigma}}{\bf{V}}^H,
\end{equation}
where ${\bf{U}}\in \mathbb{C}^{N\times N}$ and ${\bf{V}}\in
\mathbb{C}^{M\times M}$ are  both unitary matrices, and
${\bf{\Sigma}}=\left [\mathbf{\Theta}|\mathbf{0}\right ]$, with
$\mathbf{\Theta}=\mathrm{diag}\left \{ \sqrt{\theta_1},\ldots,\sqrt{\theta_N}
\right \}$ and $\mathbf{0}$ being an $N\times \left (M-N\right )$
zero matrix. Then, we propose the precoding matrix ${\bf{F}}$ at
the base station as the first $K$ columns of ${\bf{V}}$ and the
receiving matrix ${\bf{W}}={\bf{W}}_1$ at the relay as ${\bf{U}}^H$. Thus we
have
\begin{equation}
\rho_s=\sqrt{\frac{P_s}{\mathrm{tr}({\bf{F}}^H {\bf{F}})}}=\sqrt{\frac{P_s}{K}}.
\end{equation}

For the second phase, the transmission is a broadcast channel.
Instead of zero-forcing (ZF) or matched-filter (MF) in tradition, we
design a robust regularized zero-forcing (RZF) precoder for the
forward channel (FC). Given the imperfect FC matrix, the RZF at
relay is $\widehat{\bf{G}}^H (\widehat{\bf{G}} \widehat{\bf{G}}^H +
\alpha {\bf{I}}_K)^{-1}$. We aim at optimizing $\alpha$ in the RZF
precoder in terms of SNRs of BC and FC and the powers of channel
estimation errors $e_1^2$ and $e_2^2$. Since the power penalty
problem of ZF mostly exists in the case $N=K$~\cite{3}, we assume
$N=K$. Generally, a non-zero $\alpha$ will bring interference, but
can reduce the power penalty. To optimize $\alpha$, we  need to
derive the SINR in terms of $\alpha$ at each user.  In the following
we will see that $\alpha$ can be optimized based on the SINR
expressed by the eigenvalues of the instantaneous CSI at each user
terminal, and for large $K$ case, the $\alpha$ is independent of the
instantaneous CSI. For SVD-RZF, we have
\begin{eqnarray}
\mathbf{F}&=&\mathbf{V},\label{6}\\
\mathbf{W}&=&\mathbf{W}_1= \widehat{\mathbf{G}}^H\left(\widehat{\mathbf{G}}\widehat{\mathbf{G}}^H+\alpha\mathbf{I}_k\right)^{-1}\mathbf{U}^H\label{7}.
\end{eqnarray}
In the following derivation, we use the decomposition
\begin{equation}
\widehat{\mathbf{G}}{\widehat{\mathbf{G}}}^H=\mathbf{Q}\mathrm{diag}\{\lambda_1,\ldots,\lambda_K\}{\mathbf{Q}}^H.
\label{9}
\end{equation}
Substituting (\ref{7}) and (\ref{9}) into (\ref{8}), we have the power control factor respectively as
\begin{equation}
\rho_s=\left(\frac{P_s}{K}\right)^{\frac{1}{2}},
\end{equation}
and (\ref{30}) which is written at the top of next page.

Substituting (\ref{6}) and (\ref{7}) into (\ref{10}), through some manipulations, we have the power of effective noise
\begin{multline}
N(\theta,\lambda)=\left(\rho_s^2 \rho_r^2 e_1^2+\rho_r^2\frac{\sigma_1^2}{K}\right)\sum \frac{\lambda^2}{(\lambda+\alpha)^2}\\+\left(\frac{1}{K}\rho_s^2\rho_r^2e_2^2\sum\theta+\rho_r^2e_2^2\sigma_1^2\right)\sum\frac{\lambda}{(\lambda+\alpha)^2}+\sigma_2^2,
\end{multline}
where in the derivation, we have taken expectation over unitary matrix $\mathbf{Q}$.
The received data signal vector at the user terminals can be calculated as
\begin{equation}
\rho_s\rho_r\widehat{\mathbf{G}}\mathbf{W}\widehat{\mathbf{H}}\mathbf{F}\mathbf{s}=
\rho_s\rho_r\widehat{\mathbf{G}}{\widehat{\mathbf{G}}}^H\left(\widehat{\mathbf{G}}{\widehat{\mathbf{G}}}^H+\alpha\mathbf{I}_k\right)^{-1}
\mathbf{\Theta}\mathbf{s}.
\end{equation}
From the above expression, we see that the effective channel matrix is not diagonal when $\alpha$ is not zero. So the received signal by a user terminal consists of the desired signal and the interference from other users' signal.
To divides the interference from the desired signal, we introduce the following two lemmas.
\newtheorem{theorem}{Theorem}
\newtheorem{lemma}[theorem]{Lemma}
\begin{lemma}
If $\mathbf{A}=\mathbf{Q}\mathbf{\Lambda}{\mathbf{Q}}^H$, then $\mathrm{E}\left\{\left(\mathbf{A}\right)_{k,k}^2\right\}=\frac{1}{K(K+1)}\left(\left(\sum\lambda\right)^2+\sum \lambda^2\right)\triangleq\mu(\lambda)$.
\end{lemma}
\begin{figure*}[!t]
\begin{equation}\label{30}
\rho_r=\left(\frac{P_r}{\mathrm{tr}\left(\widehat{\mathbf{G}}{\widehat{\mathbf{G}}}^H\left(\widehat{\mathbf{G}}{\widehat{\mathbf{G}}}^H+\alpha\mathbf{I}_k\right)^{-2}\left(\rho_s^2{\mathbf{\Theta}}^2+\rho_s^2e_1^2\mathbf{\Omega}_1\mathbf{\Omega}_1^H+\sigma_1^2\mathbf{I}_k\right)\right)}\right)^{\frac{1}{2}}
=\left(\frac{ P_r}{(\frac{P_s}{K^2} \sum \theta+e_1^2 P_s+\sigma_1^2 )\sum \frac{\lambda}{(\lambda+\alpha)^2}}\right)^{\frac{1}{2}}.
\end{equation}
\begin{equation}\label{1}
\mathrm{SINR}\overset{w.p.}{\longrightarrow}\frac{\frac{P_s}{M}\left(R{\mathcal{E}}_1^\theta\mathcal{E}_1^\lambda\right)^2}
{\frac{P_sR(M-1)}{M^2}{\mathcal{E}}_3^\theta{\mathcal{E}}_3^\lambda+\left(e_1^2P_s+\sigma_1^2\right)R{\mathcal{E}}_2^\theta{\mathcal{E}}_3^\lambda
+P_sRe_2^2{\mathcal{E}}_3^\theta{\mathcal{E}}_2^\lambda
+e_2^2\sigma_1^2RM{\mathcal{E}}_2^\theta{\mathcal{E}}_2^\lambda
+\sigma_2^2\rho_r^{-2}}.
\end{equation}
\hrulefill
\end{figure*}
The proof of Lemma 1 can be directly obtained in~\cite{3}.
\begin{lemma}
If $\mathbf{A}=\mathbf{Q}\mathbf{\Lambda}{\mathbf{Q}}^H$, then $\mathrm{E}\left\{\left(\mathbf{A}\right)_{k,j}^2\right\}=\frac{1}{(K-1)(K+1)}\sum\lambda^2-\frac{1}{(K-1)K(K+1)}\left(\sum \lambda\right)^2\triangleq\nu(\lambda)$, for $k\neq j$.
\end{lemma}
\begin{proof}
Because $\mathbf{A}$ is a conjugate symmetric matrix, we have
\begin{multline}
\mathrm{E}\left\{\sum_{j=1,j\neq k}^K
|\left(\mathbf{A}\right)_{k,j}|^2\right\}+\mathrm{E}\left\{\left(\mathbf{A}\right)_{k,k}^2\right\}
=\mathrm{E}\left\{\left(\mathbf{A}\mathbf{A}^H\right)_{k,k}\right\}
\\=\mathrm{E}\left\{\left(\mathbf{Q}\mathbf{\Lambda}^2\mathbf{Q}^H\right)_{k,k}\right\}=\frac{1}{K}\sum\lambda^2.
\end{multline}
Since $\mathrm{E}\left\{\left(\mathbf{A}\right)_{k,j}^2\right\}$ are all equal for $j\neq k$, we have
\begin{multline}
\mathrm{E}\left\{|\left(\mathbf{A}\right)_{k,j}|^2\right\}=\frac{1}{(K-1)}\left(\frac{1}{K}\sum\lambda^2-\mathrm{E}\{\left(\mathbf{A}\right)_{k,k}^2\}\right)
\\=\frac{1}{(K-1)(K+1)}\sum\lambda^2-\frac{1}{(K-1)K(K+1)}\left(\sum \lambda\right)^2.
\end{multline}
\end{proof}
Therefore, for user-$k$, if we denote $\mathbf{A}=\widehat{\mathbf{G}}{\widehat{\mathbf{G}}}^H\left(\widehat{\mathbf{G}}{\widehat{\mathbf{G}}}^H+\alpha\mathbf{I}_k\right)^{-1}$, we can calculate the power of desired signal as
\begin{equation}
\mathrm{E}\left\{\left\|\mathbf{A}_{k,k}\theta_k
\left(\mathbf{s}\right)_k\right\|^2\right\}=\rho_s^2\rho_r^2
\theta_k \mu\left(\frac{\lambda}{\lambda+\alpha}\right).
\end{equation}
The power of interference is
\begin{equation}
\mathrm{E}\left\{\left\|\sum_{j=1,j\neq k}^K\mathbf{A}_{k,j}\theta_j
\left(\mathbf{s}\right)_j\right\|^2\right\}=\rho_s^2\rho_r^2
\left(\sum_{j=1,j\neq k}^K\theta_j\right)
\nu\left(\frac{\lambda}{\lambda+\alpha}\right).
\end{equation}
Finally, the SINR at user-$k$ is
\begin{equation}
\mathrm{SINR}_k=\frac{\rho_s^2\rho_r^2 \theta_k \mu(\frac{\lambda}{\lambda+\alpha})}
{\rho_s^2\rho_r^2 \left(\sum_{j=1,j\neq k}^K\theta_j\right) \nu(\frac{\lambda}{\lambda+\alpha})+N(\theta,\lambda)}.
\end{equation}
Note that in the above expression, the SINR is based on the eigenvalue of instantaneous imperfect CSIs.
To maximize the SINR expression, we introduce the following lemma which is a conclusion of the Appendix B in~\cite{3}.
\begin{lemma}
\begin{equation}
\mathrm{SINR}(\alpha)=\frac{A\left(\sum\frac{\lambda}{\lambda+\alpha}\right)^2+B\sum\frac{\lambda^2}{(\lambda+\alpha)^2}}{C\sum\frac{\lambda}{(\lambda+\alpha)^2}+D\sum\frac{\lambda^2}{(\lambda+\alpha)^2}+E\left(\sum\frac{\lambda}{\lambda+\alpha}\right)^2},
\end{equation}
for large $K$, is maximized by $\alpha=C/D$.
\end{lemma}

Using Lemma 3, we finally get the optimized
\begin{equation}
\alpha_{SVD-RZF,\mathrm{opt}}=\frac{\frac{e_2^2 \sum\theta}{K}+\frac{e_2^2\sigma_1^2K}{P_s}+\frac{\sigma_2^2}{P_r}(\frac{\sum\theta}{K}+Ke_1^2+\frac{\sigma_1^2K}{P_s})}{\frac{\sum\theta_j}{(K-1)(K+1)}+e_1^2+\frac{\sigma_1^2}{P_s}}.
\end{equation}
For large $K$, we have
\begin{equation}
\begin{split}
\alpha_{SVD-RZF,\mathrm{opt}}&\approx
\frac{\frac{e_2^2 K^2}{K}+\frac{e_2^2\sigma_1^2K}{P_s}+\frac{\sigma_2^2}{P_r}(\frac{K^2}{K}+Ke_1^2+\frac{\sigma_1^2K}{P_s})}{\frac{K(K-1)}{(K-1)(K+1)}+e_1^2+\frac{\sigma_1^2}{P_s}}
\\&\approx K \left(\frac{e_2^2+e_2^2\sigma_1^2/P_s}{1+e_1^2+\sigma_1^2/P_s}+\sigma_2^2/P_r\right).
\end{split}
\end{equation}
\subsection{MMSE-RZF based design for multi-relay case}
Although SVD is advantageous, it can only be implemented in the single relay case. For the multi-relay case, the relays have to work in a cooperative mode to diagonalize the channel as SVD or the base station needs the CSI of all the backward channnels which will lead to considerable delay. Therefore, for the multi-relay case, we propose another beamforming scheme which is based on MMSE-RZF instead of SVD-RZF.

It is known that MMSE receiver is widely used in point-to-point MIMO
systems.  The MMSE receiver can be viewed as a duality of the RZF
precoder, where the difference is that the RZF precoder is
frequently used in multiantenna multiuser communication. Our main
idea is to obtain the optimal regularizing factor in MMSE receiver
to reduce the effect of channel estimation error of the backward
channels.

The MMSE receiver at the $r$-th relay is
$\left(\widehat{\bf{H}}_r^H\widehat{\bf{H}}_r+\alpha^{\mathrm{MMSE}}\right)^{-1}\widehat{\bf{H}}_r^H$.
For the same reason as RZF, MMSE receiver is most superior to other
linear receivers (e.g. ZF) when $M=N$. So we consider $M=N=K$ for
the multi-relay case. Because the aim of MMSE receiver is to reduce
the effect of channel estimation error of BC, we optimize
$\alpha^{\mathrm{MMSE}}$ by idealizing the forward channels as
Gaussian channels, i.e., the forward channel is considered as
$\widehat{\bf{G}}_r={\bf{G}}_r={\bf{I}}_N$.

In the following analysis, we use the decompositions,
\begin{eqnarray}
\widehat{\bf{H}}_r^H\widehat{\bf{H}}_r&=&\mathbf{P}_r\mathrm{diag}\{\theta_{r,1},\ldots,\theta_{r,N}\}\mathbf{P}_r^H,\\
\widehat{\bf{G}}_r\widehat{\bf{G}}_r^H&=&\mathbf{Q}_r\mathrm{diag}\{\lambda_{r,1},\ldots,\lambda_{r,N}\}\mathbf{Q}_r^H,
\end{eqnarray}
where $\mathbf{P}_r$ and $\mathbf{Q}_r$ are unitary matrices. For the $r$-th relay, the signal vector processed by an MMSE receiver is
\begin{equation}
\begin{split}
{\bf{v}}_r  =& \rho_r\left({\widehat{\bf{H}}}_r^H {\widehat{\bf{H}}}_r+ \alpha^{\mathrm{MMSE}} {\bf{I}}_M \right)^{-1}{\widehat{\bf{H}}}_r^H {\bf{r}}_r\\=&\rho_r\left({\widehat{\bf{H}}}_r^H {\widehat{\bf{H}}}_r+ \alpha^{\mathrm{MMSE}} {\bf{I}}_M \right)^{-1}{\widehat{\bf{H}}}_r^H {\widehat{\bf{H}}}_r {\bf{s}}\\& + \rho_r e_1\left({\widehat{\bf{H}}}_r^H {\widehat{\bf{H}}}_r+ \alpha^{\mathrm{MMSE}} {\bf{I}}_M \right)^{-1}{\widehat{\bf{H}}}_r^H {\bf{\Omega}}_{1,r}
{\bf{s}}\\&+\rho_r\left({\widehat{\bf{H}}}_r^H {\widehat{\bf{H}}}_r+ \alpha^{\mathrm{MMSE}} {\bf{I}}_M \right)^{-1}{\widehat{\bf{H}}}_r^H{\bf{n}}_r.
\end{split}
\end{equation}
Using similar manipulations with the single relay case, the SINR of the $k$-th user's data at the $r$-th relay is
\begin{equation}
\mathrm{SINR}_{r,k}^{\mathrm{R}}=\frac{\frac{P_s}{M}\mu\left(\frac{\theta_r}{\theta_r+\alpha^{\mathrm{MMSE}}}\right)}
{\frac{P_s(M-1)}{M}\nu\left(\frac{\theta_r}{\theta_r+\alpha^{\mathrm{MMSE}}}\right)+\frac{e_1^2P_s+\sigma_1^2}{M}\sum\frac{\theta_r}{\left(\theta_r+\alpha^{\mathrm{MMSE}}\right)^2}}.
\end{equation}
At the destination, the received vector is from all the $R$ relays.
So the desired signal is scaled by $R^2$ and the interference and the noise inherited from the relays are scaled by $R$.
Therefore, by idealizing the forward channels, we have the SINR of the $k$-th stream as
\begin{equation}\label{19}
\begin{split}
&\mathrm{SINR}_{k}^{\mathrm{D}}\\\approx& \frac{\frac{P_sR^2}{M}\mu\left(\frac{\theta}{\theta+\alpha^{\mathrm{MMSE}}}\right)}
{R\frac{P_s(M-1)}{M}\nu\left(\frac{\theta}{\theta+\alpha^{\mathrm{MMSE}}}\right)+R\frac{e_1^2P_s+\sigma_1^2}{M}\sum\frac{\theta}{\left(\theta+\alpha^{\mathrm{MMSE}}\right)^2}+\sigma_2^2\rho_r^{-2}},
\end{split}
\end{equation}
where the power control factor $\rho_r$ at relay normalizes the noise at the destination.
We use the same $\rho_r$ for all the relays for the simplicity of analysis by taking expectation to the denominator in (\ref{8}).
Using Lemma 3, we obtain
\begin{equation}\label{29}
\begin{split}
\alpha^{\mathrm{MMSE,opt}}=&\frac{\frac{P_s e_1^2+\sigma_1^2}{P_r}\sigma_2^2+R\frac{P_se_1^2+\sigma_1^2}{M}}{\frac{P_s\sigma_2^2}{MP_r}+R\frac{P_s(M-1)}{M}\frac{1}{(M-1)(M+1)}}
\\=&\left(e_1^2+\frac{\sigma_1^2}{P_s}\right)\frac{M+\frac{P_rR}{\sigma_2^2}}{1+\frac{P_rR}{(M+1)\sigma_2^2}}.
\end{split}
\end{equation}

To obtain the optimal $\alpha^{\mathrm{RZF}}$, we need to derive the asymptotic SINR of the system. Again, we separate the desired signals from the interference and the noise and finally derive the SNR at the $k$-th user terminal as
\begin{equation}\label{13}
\mathrm{SINR}_k^{\mathrm{D}}=\frac{\frac{P_s}{M}|(\mathbf{H}_\mathcal{SD})_{k,k}|^2}
{\frac{P_s}{M}\sum_{j=1,j\neq k}^K|(\mathbf{H}_\mathcal{SD})_{j,k}|^2+N\left({\bf{G}}_r,{\bf{H}}_r\right) },
\end{equation}
where
\begin{equation}
\mathbf{H}_\mathcal{SD}=\sum_{r=1}^R\rho_r{\widehat{\bf{G}}}_r {\bf{W}}_r{\widehat{\bf{H}}}_r,
\end{equation}
and
\begin{multline}\label{36}
N\left({\bf{G}}_r,{\bf{H}}_r\right)
=\left(e_1^2P_s+\sigma_1^2\right)\sum_{r=1}^R\left\|\rho_r\left(\widehat{{\bf{G}}}_r{\bf{W}}_r\right)_k\right\|^2
\\+\frac{P_se_2^2}{M}\sum_{r=1}^R\rho_r^2\mathrm{tr}\left(\mathbf{W}_r\widehat{\mathbf{H}}_r\widehat{\mathbf{H}}_r^H\mathbf{W}_r^H\right)
\\+e_2^2\sigma_1^2\sum_{r=1}^R\rho_r^2\mathrm{tr}\left(\mathbf{W}_r\mathbf{W}_r^H\right)+\sigma_2^2.
\end{multline}
For the case of large $R$, using Law of Large Number, we have
\begin{equation}\label{14}
\begin{split}
&|\left({\bf{H}}_{{\mathcal{S}\mathcal {D}}}\right)_{i,i}|\overset{w.p.}{\longrightarrow}R\rho_r\left(\mathrm{E}\left\{\left({\widehat{\bf{G}}}_r {\bf{W}}_r{\widehat{\bf{H}}}_r\right)_{i,i}\right\}\right)\\
=&R\rho_r\mathrm{E}\biggl\{\left(\mathbf{Q}_r\frac{\mathbf{\Lambda}_r}{\mathbf{\Lambda}_r+\alpha^{\mathrm{RZF}}\mathbf{I}_M}\mathbf{Q}_r^H\mathbf{P}_r\frac{\mathbf{\Theta}_r}{\mathbf{\Theta}_r+\alpha^{\mathrm{MMSE}}\mathbf{I}_N}\mathbf{P}_r^H\right)_{i,i}\biggl\}
\\=&R\rho_r\mathrm{E}\left\{\left(\mathbf{Q}_r\frac{\mathbf{\Lambda}_r}{\mathbf{\Lambda}_r+\alpha^{\mathrm{RZF}}\mathbf{I}_M}\mathbf{Q}_r^H\right)_{m,m}\right\}
\\&\quad\quad\quad\quad\quad\quad\quad\quad\quad\quad\mathrm{E}\left\{\left(\mathbf{P}_r\frac{\mathbf{\Theta}_k}{\mathbf{\Theta}_r+\alpha^{\mathrm{MMSE}}\mathbf{I}_N}\mathbf{P}_r^H\right)_{n,n}\right\}
\\=&\frac{R\rho_r}{MN}\mathrm{E}\left\{\sum\frac{\theta_r}{\theta_r+\alpha^{\mathrm{MMSE}}}\right\}\mathrm{E}\left\{\sum\frac{\lambda_r}{\lambda_r+\alpha^{\mathrm{RZF}}}\right\}
\\=&R\rho\mathrm{E}\left\{\frac{\theta}{\theta+\alpha^{\mathrm{MMSE}}}\right\}\mathrm{E}\left\{\frac{\lambda}{\lambda+\alpha^{\mathrm{RZF}}}\right\},
\end{split}
\end{equation}
and
\begin{equation}\label{15}
\begin{split}
&\quad\left|({\bf{H}}_{{\mathcal{S}\mathcal {D}}})_{(i,j)}\right|^2\\
&=\left|\left(\sum_{r=1}^R\mathbf{Q}_r\frac{\mathbf{\Lambda}_r}{\mathbf{\Lambda}_r+\alpha^{\mathrm{RZF}}\mathbf{I}_M}\mathbf{Q}_r^H\mathbf{P}_r\frac{\mathbf{\Theta}_r}{\mathbf{\Theta}_r+\alpha^{\mathrm{MMSE}}\mathbf{I}_N}\mathbf{P}_r^H\right)_{(i,j)}\right|^2
\\
&=\sum_r\left|(\mathbf{Q}_r)_{i,k}\frac{\lambda_{r,k}}{\lambda_{r,k}+\alpha^{\mathrm{MMSE}}}
(\mathbf{Q}_r)_{l,k}^*\right.\\&\quad\quad\quad\quad\quad\quad\quad\quad\quad\quad\quad\quad\left.(\mathbf{P}_r)_{l,m}\frac{\theta_{r,m}}{\theta_{r,m}+\alpha^{\mathrm{RZF}}}(\mathbf{Q}_r)_{j,m}^*\right|^2
\\&\overset{(a)}{\approx}\sum_{k,m,n,r}\frac{1}{M^4}\left(\frac{\lambda_{r,k}}{\lambda_{r,k}+\alpha^{\mathrm{RZF}}}\right)^2\left(\frac{\theta_{r,k}}{\theta_{r,k}+\alpha^{\mathrm{MMSE}}}\right)^2
\\&\overset{w.p.}{\longrightarrow}\frac{R}{M}\mathrm{E}\left\{\frac{\theta^2}{(\theta+\alpha^{\mathrm{MMSE}})^2}\right\}
\mathrm{E}\left\{\frac{\lambda^2}{(\lambda+\alpha^{\mathrm{RZF}})^2}\right\}
\end{split}
\end{equation}
where in (a) we approximate $\mathrm{E}\left\{|(\mathbf{Q}_r)_{i,k}|^2|(\mathbf{Q}_r)_{l,k}|^2\right\}\approx \frac{1}{M^2}$. In fact, this expectation is $\frac{2}{M(M+1)}$ if $i=l$ or $\frac{1}{M(M+1)}$ if $i\neq l$~\cite{3}.
Here we denote $\lambda$ and $\theta$ without subscript $r$ for simplicity, because all the channels for different relays are $i.i.d$.
Let us define the expectations as $
{\mathcal{E}}_1^\theta\triangleq\mathrm{E}\left\{\frac{\theta}{(\theta+\alpha^{\mathrm{MMSE}})}\right\},
{\mathcal{E}}_2^\theta\triangleq\mathrm{E}\left\{\frac{\theta}{(\theta+\alpha^{\mathrm{MMSE}})^2}\right\},
{\mathcal{E}}_3^\theta\triangleq\mathrm{E}\left\{\frac{\theta^2}{(\theta+\alpha^{\mathrm{MMSE}})^2}\right\},
{\mathcal{E}}_1^\lambda\triangleq\mathrm{E}\left\{\frac{\lambda}{(\lambda+\alpha^{\mathrm{RZF}})}\right\},
{\mathcal{E}}_2^\lambda\triangleq\mathrm{E}\left\{\frac{\lambda}{(\lambda+\alpha^{\mathrm{RZF}})^2}\right\},
{\mathcal{E}}_3^\lambda\triangleq\mathrm{E}\left\{\frac{\lambda^2}{(\lambda+\alpha^{\mathrm{RZF}})^2}\right\}.
$ Substituting (\ref{36})-(\ref{15}) into (\ref{13}), we obtain the
asymptotic SINR at each user terminal as (\ref{1}) at the top of the last page,
where
\begin{equation}
\begin{split}
&\rho_r^{-2}
=\frac{1}{P_r}\times\\&\mathrm{E}\left\{\frac{P_s}{M}\mathrm{tr}\left(\mathbf{F}_k(\widehat{\mathbf{H}}_k\widehat{\mathbf{H}}_k^H+e_1^2\mathbf{\Omega}_{1,k}\mathbf{\Omega}_{1,k}^H)\mathbf{F}_k^H\right)+\sigma_1^2\mathrm{tr}\left(\mathbf{F}_k\mathbf{F}_k^H\right)\right\}
\\=&\frac{P_s}{P_r}{\mathcal{E}}_3^\theta{\mathcal{E}}_2^\lambda+\frac{(e_1^2P_s+\sigma_1^2)M}{P_r}{\mathcal{E}}_2^\theta{\mathcal{E}}_2^\lambda.
\end{split}
\end{equation}
The calculation of (\ref{36}) can follow the same line as
(\ref{14}). Generally, the expectations in the asymptotic SINR are
difficult. Fortunately, if we approximate the expectations by the
arithmetic mean,
for large $R$, then the asymptotic SINR can be maximized by using Lemma 3. Finally, we obtain
\begin{equation}\label{24}
\alpha^{\mathrm{RZF},\mathrm{opt}}\approx\frac{(P_sRe_2^2+\frac{\sigma_2^2P_s}{P_r})\mathcal{E}_3^\theta+(e_2^2\sigma_1^2RM+\frac{(e_1^2P_s+\sigma_1^2)M}{P_r})\mathcal{E}_2^\theta}
{(e_1^2P_s+\sigma_1^2)R\mathcal{E}_2^\theta+\frac{P_s R}{M}\mathcal{E}_3^\theta}.
\end{equation}

Note that although we maximize the SINR for large $K$, and large $R$ for multi-relay case, we will see from the numerical simulation that the obtained beamforming is robust enough for small $K$ and $R$ when channel estimation error occurs.

\section{Simulation results}
In this section, numerical simulations have been carried out.
For the single relay case, we compare the SINR at each user terminal
of the robust SVD-RZF beamforming with SVD-ZF and SVD-MF
in~\cite{21}, MMSE-RZF in~\cite{26}, and two other relative
beamforming schemes such as ZF-ZF and SVD-RZF for references. For
MMSE-RZF, $\alpha_{\mathrm{MMSE}}=K\sigma_1^2/P_s$, and
$\alpha_{\mathrm{RZF}}=K\sigma_2^2/P_r$. We also consider the robust
MMSE-RZF proposed for multi-relay case for $R=1$. For the
multi-relay case, we compare with the conventional MMSE-RZF and
ZF-ZF. All the results are averaged over 10000 different channel
realizations.
\begin{figure}[!t]
\centering
\includegraphics[width=3.5in]{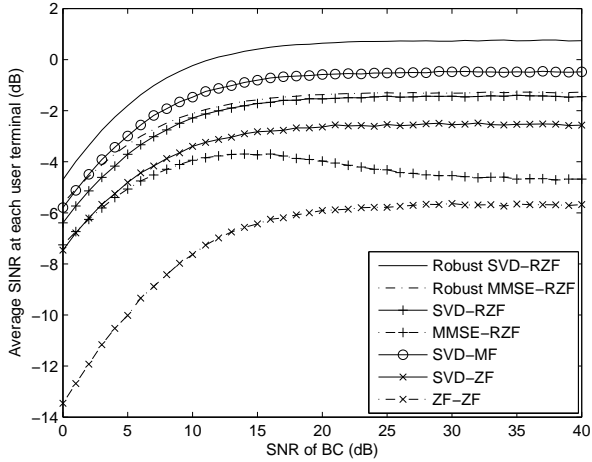}
\caption{The SINRs at each user terminal for different beamforming schemes
versus the SNR of BC in the single relay case. $P_r/\sigma_2^2=20dB$, $e_1^2=0.2$ and $e_2^2=0.2$. The robust MMSE-RZF and MMSE-RZF will change with the SNR of BC due to the regularizing factor in MMSE receiver.}
\end{figure}

\begin{figure}[!t]
\centering
\includegraphics[width=3.5in]{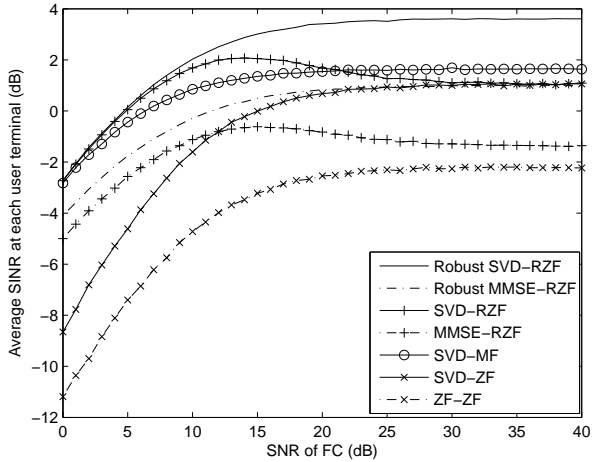}
\caption{The SINRs at each user terminal for different beamforming schemes
versus the SNR of FC in the single relay case. $P_s/\sigma_1^2=20dB$, $e_1^2=0.1$ and $e_2^2=0.1$. The beamformings with RZF will change with the SNR of FC due to the regularizing factor.}
\end{figure}
\begin{figure}[!t]
\centering
\includegraphics[width=3.5in]{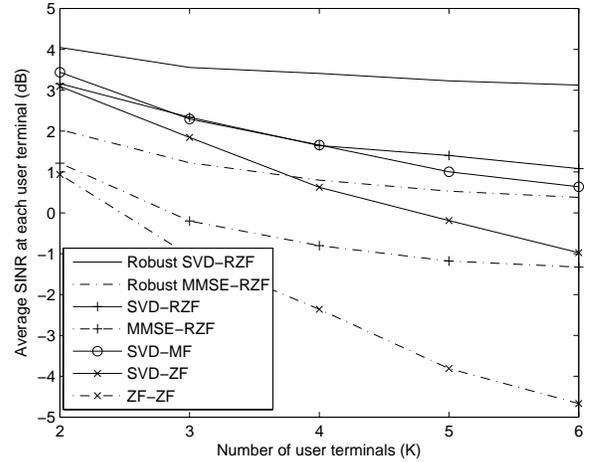}
\caption{The SINRs at each user terminal for different beamforming schemes
versus the number of users ($K$) in the single relay case. $P_s/\sigma_1^2=20dB$, $P_r/\sigma_2^2=20dB$, $e_1^2=e_2^2=0.1$.}
\end{figure}

\begin{figure}[!t]
\centering
\includegraphics[width=3.5in]{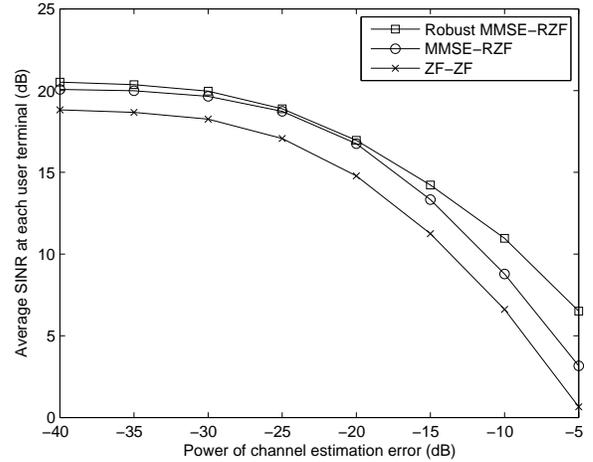}
\caption{The SINRs at each user terminal for different beamforming schemes
versus the power of channel estimation error in the multi-relay case. $P_s/\sigma_1^2=P_r/\sigma_2^2=20dB$, $e_1=e_2$, $R=10$. The power of channel estimation error is $e_1^2=e_2^2$.}
\end{figure}
\subsection{SINR performances for the single relay case}
Fig.~2 shows the SINRs of different
beamforming schemes  versus the SNR of BC. We observe that the
proposed robust SVD-RZF beamforming has consistently advantage to
others.
Robust MMSE-RZF underperforms robust SVD-RZF and SVD-MF, which shows the superior of SVD.
Fig.~3 shows the SINRs versus the SNR of FC.   The SINR of SVD-RZF
even falls and converges to SVD-ZF when the SNR of FC increases,
because the $\alpha$ converges to zero, which should remain nonzero
if estimation error is considered. Fig.~4 shows the SINRs versus the
number of users ($K$). We see that the robust SVD-RZF also
outperforms others when $K$ is small. The advantage of robust
SVD-RZF comes from the fact that the SVD beamforming  outperforms
robust MMSE receiver although the former ignores the estimation
error. For the broadcast phase, the robust RZF compensates well the
estimation error compared to ZF and RZF.
\begin{figure}[!t]
\centering
\includegraphics[width=3.5in]{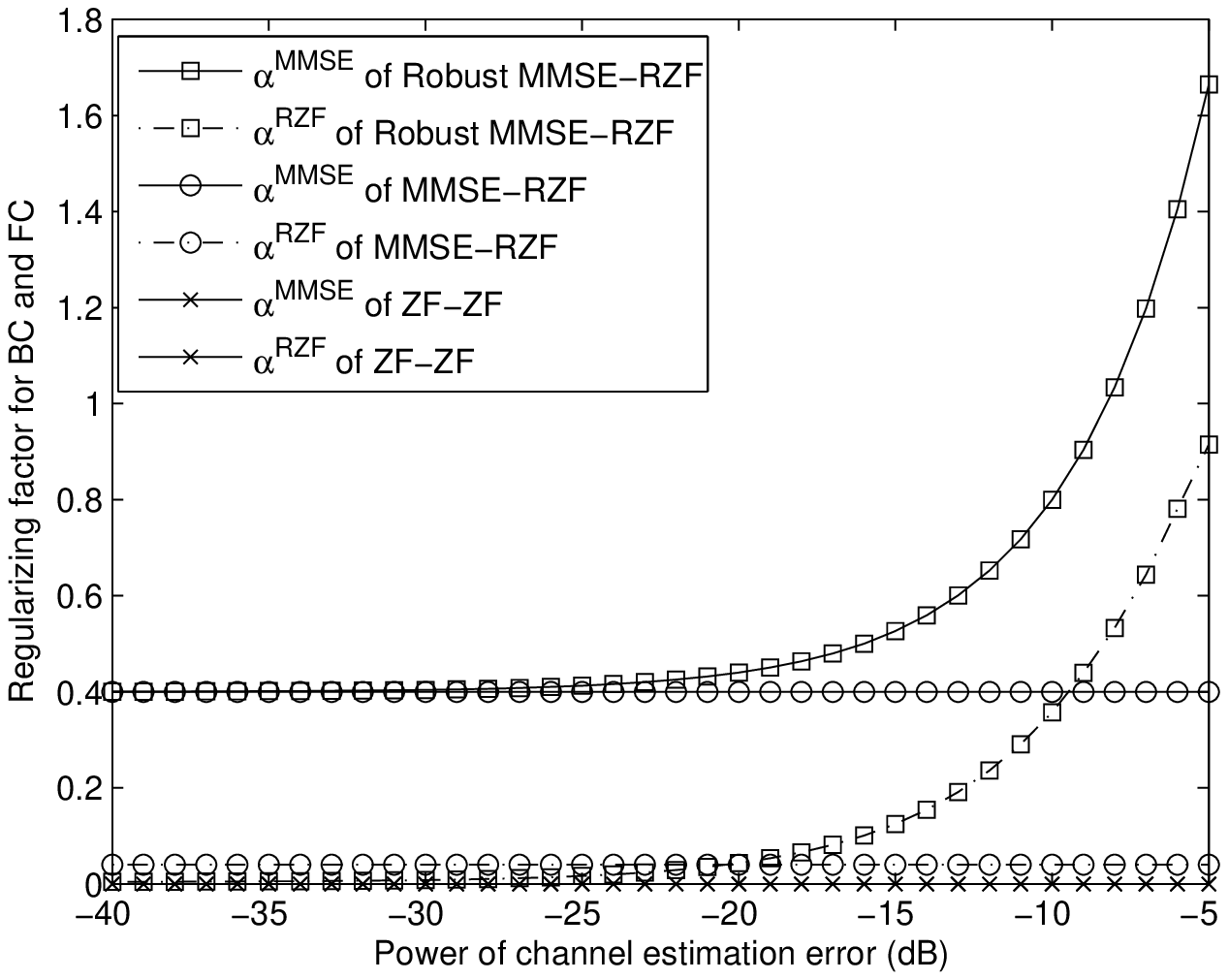}
\caption{$\alpha^{\mathrm{MMSE}}$ and $\alpha^{\mathrm{RZF}}$ for different beamforming schemes
versus the power of channel estimation error in the multi-relay case. $P_s/\sigma_1^2=10dB$, $P_r/\sigma_2^2=20dB$, $e_1=e_2$, $R=10$.}
\end{figure}

\begin{figure}[!t]
\centering
\includegraphics[width=3.5in]{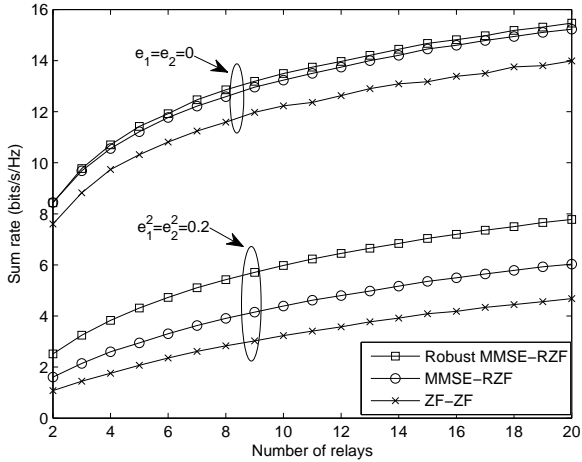}
\caption{The sum rates for different beamforming schemes
versus the number of relays ($R$) in the multi-relay case with perfect channel estimation and channel estimation errors. $P_s/\sigma_1^2=P_r/\sigma_2^2=20dB$, $e_1=e_2=0$ or $e_1^2=e_2^2=0.2$. The sum rate are averaged by $0.5K\mathrm{log}_2\left(1+\mathrm{SINR}_k\right)$. The factor $0.5$ is due to the two time slots transmission.}
\end{figure}
\subsection{SINR performances for the multi-relay case}
For multi-relay case where SVD can not be implemented, we only compare the proposed robust MMSE-RZF with MMSE-RZF and ZF-ZF.
Fig.~5 shows the average SINR performances versus the power of channel estimation error ($e_1^2=e_2^2$). This is because that the $\alpha^{\mathrm{MMSE}}$ and $\alpha^{\mathrm{RZF}}$ increase with $e_1$ and $e_2$ to decrease the effect of estimation error. This can be directly seen from Fig.~6.
Fig.~7 shows the sum rate performances versus the number of relays ($R$) with perfect and imperfect channel estimation. We see that all sum rates grows logarithmically with $R$ and the superior of robust MMSE-RZF increases when channel estimation is imperfect or the number of relays grows. This is because that comparing with conventional MMSE-RZF, the robust one considers both imperfect channel estimation and multiple relays.

\section{Conclusion}
In this paper we propose the robust SVD-RZF and robust MMSE-RZF
beamformers which consider imperfect channel estimation  for a
multiuser downlink MIMO relaying network. For the single relay case,
the SINR expression at user terminals based on the eigenvalue of BC
and FC matrix is  derived to obtain the optimized RZF. For the
multi-relay case, the asymptotic SINR is derived to obtain the
optimized MMSE and RZF. Simulation results show that the proposed
robust SVD-RZF and MMSE-RZF outperform the conventional schemes for
various  conditions of SNR of channels, power of estimation errors,
the number of antennas, users and the relays.

\appendices

\end{document}